\documentclass[12pt]{article}
\usepackage{amsfonts,amsthm,amsmath,amssymb,mathrsfs,url,graphicx}
\usepackage[usenames,dvipsnames]{color}

\title{Epistemology of Wave Function Collapse\\ in Quantum Physics}

\author{
Charles Wesley Cowan\footnote{Department of Mathematics,
     Rutgers University, Hill Center,  
     110 Frelinghuysen Road, Piscataway, NJ 08854-8019, USA.}\ \footnote{E-mail:
     cwcowan@math.rutgers.edu}\ \ and 
Roderich Tumulka$^*$\footnote{E-mail: tumulka@math.rutgers.edu}
}

\date{February 19, 2014}

\newcommand{\Hilbert}{\mathscr{H}}
\newcommand{\EEE}{\mathbb{E}}
\newcommand{\PPP}{\mathbb{P}}
\newcommand{\RRR}{\mathbb{R}}

\newcommand{\SSS}{\mathbb{S}}

\newcommand{\scp}[2]{\langle #1|#2 \rangle}
\newcommand{\pr}[1]{| #1 \rangle \langle #1 |}
\newcommand{\be}{\begin{equation}}
\newcommand{\ee}{\end{equation}}

\newcommand{\yes}{\mathrm{yes}}
\newcommand{\no}{\mathrm{no}}
\newcommand{\vQ}{Q} 
\newcommand{\vu}{u} 
\newcommand{\E}{\mathscr{E}}
\newcommand{\sys}{\mathrm{sys}}
\newcommand{\mm}{\widetilde{m}}
\newcommand{\MM}{\widetilde{M}}
\DeclareMathOperator{\tr}{tr}
\DeclareMathOperator{\diag}{diag}
\theoremstyle{plain}\newtheorem{prop}{Proposition}
\newcommand{\R}{\mathscr{R}}
\newcommand{\N}{\mathscr{N}}

\begin{document}
\maketitle
\begin{abstract}
Among several possibilities for what reality could be like in view of the empirical facts of quantum mechanics, one is provided by theories of spontaneous wave function collapse, the best known of which is the Ghirardi--Rimini--Weber (GRW) theory. We show mathematically that in GRW theory (and similar theories) there are limitations to knowledge, that is, inhabitants of a GRW universe cannot find out all the facts true about their universe. As a specific example, they cannot accurately measure the number of collapses that a given physical system undergoes during a given time interval; in fact, they cannot reliably measure whether one or zero collapses occur. Put differently, in a GRW universe certain meaningful, factual questions are empirically undecidable. We discuss several types of limitations to knowledge and compare them with those in other (no-collapse) versions of quantum mechanics, such as Bohmian mechanics. Most of our results also apply to observer-induced collapses as in orthodox quantum mechanics (as opposed to the spontaneous collapses of GRW theory).

\medskip

\end{abstract}
\newpage

\tableofcontents

\bigskip

\noindent\textit{Science may set limits to knowledge, but should not set limits to imagination.}\\
\hfill (Bertrand Russell, 1872--1970)

\section{Introduction}

Since science provides us with knowledge, it may seem surprising that it sometimes sets limitations to knowledge. By a ``limitation to knowledge'' we mean that certain facts about the world cannot be discovered or confirmed in an empirical way, no matter how big our effort, including possible future technological advances. For example, a limitation to knowledge is in place if a quantity cannot be measured although it has a well-defined value. A limitation to knowledge means that there is a fact, and we cannot know what it is, nor even guess with much of a chance of guessing correctly. Nature knows and we do not.

In this paper, we discuss certain limitations to knowledge concerning the collapse of the wave function in quantum physics. Specifically, we investigate limitations to measuring whether or not a collapse has occurred. Our results are \emph{epistemology} in the sense that they concern the possibility of a particular type of knowledge. Since they can be proved via mathematical theorems, they can be said to fall into the field of mathematical epistemology. Some preliminary results have been reported in \cite{grw3A}.

The very idea of a limitation to knowledge may seem to go against the principles of science. If there is no way of measuring a quantity $X$, then this may suggest that $X$ does not actually have a well-defined value, i.e., that nature does not know either what $X$ is. For example, in the early days of relativity theory, Lorentz and Fitzgerald proposed that the ether causes a length contraction of moving objects, which implies that the speed of earth relative to the ether cannot be measured; however, this situation suggests, as argued by Einstein, that the ``speed relative to the ether'' is not well defined. 

However, the existence of limitations to knowledge is a fact, as it is a simple consequence of quantum mechanics, independently of which interpretation of quantum mechanics we prefer. For example, 
suppose Alice prepares an ensemble of quantum systems, each with a pure state chosen randomly with distribution $\mu_1$ over the unit sphere
\be\label{SSSdef}
\SSS(\Hilbert)=\bigl\{ \psi\in\Hilbert:\|\psi\|=1\bigr\}
\ee
in Hilbert space $\Hilbert$. Suppose further that $\mu_2\neq \mu_1$ is another distribution over $\SSS(\Hilbert)$ with the same density matrix, $\rho_{\mu_1}=\rho_{\mu_2}$, where
\be\label{rhomudef}
\rho_\mu = \int\limits_{\SSS(\Hilbert)} \!\! \mu(d\psi) \, \pr{\psi}\,.
\ee
Then Bob is unable to determine (even probabilistically) by means of experiments on the systems whether Alice used $\mu_1$ or $\mu_2$ (see Appendix~\ref{sec:proofs} for the proof), while there is a fact about whether the states actually have distribution $\mu_1$ or $\mu_2$, as Alice knows the pure state of each system. Thus, the predictions of quantum mechanics imply that there are facts in the world which cannot be discovered empirically.\footnote{This argument has a curious feature that is worth commenting on. While its goal is to show that there is some information that nature knows but no observer can obtain, it actually describes a situation in which Alice (who is an observer, one would guess) is in possession of the relevant information, and only Bob cannot obtain it---except perhaps by spying out Alice's notebook! So it may seem that the argument cannot reach its goal. However, the argument does show that there is no way to obtain a certain information about a system from interaction with the system, and that is enough for our purposes. Alice then plays the role of proving (to a positivist) that that information objectively exists. Alternatively, one could also argue as follows. If Alice destroyed the records of the states she prepared then nobody would know, while it seems plausible that nature still knows the states (since it seems implausible that the states would suddenly become indefinite if Alice burns her notebook). Thus, the example also easily provides a case in which nobody is in possession of the relevant information.}

What is upsetting about limitations to knowledge is that they conflict with key ideas of (what may be called) \emph{positivism}: That a statement is unscientific or even meaningless if it cannot be tested experimentally, that an object is not real if it cannot be observed, and that a variable is not well-defined if it cannot be measured. We conclude that this form of positivism is exaggerated; it is inadequate. The above example with Alice and Bob refutes it. As an even simpler example against this form of positivism, suppose a space ship falls into a black hole; it seems reasonable to believe that it continues to exist for a while although we cannot observe it any more. While perhaps nobody would defend positivism as we described it, it is frequently applied in physics reasoning, particularly in quantum physics---ironically so, since its inadequacy is particularly clear in quantum physics, as the example with Alice and Bob shows. Needless to say, positivism is also \emph{unnecessary} in quantum physics, as demonstrated by the viability of ``realist'' interpretations of quantum mechanics (``quantum theories without observers'') such as Bohmian mechanics \cite{Bohm52,Gol01b}, spontaneous collapse theories \cite{GRW86,Bell87,Pe89,Dio89,Ghi07}, and perhaps (some versions of) many-worlds \cite{Eve57,Vai02,AGTZ11}. 

It is sometimes suggested that a theory is unconvincing if it entails limitations to knowledge. We do not share this sentiment. On the contrary, we think it is hard to defend in view of the considerations around \eqref{rhomudef}.

In this paper, we deal primarily with spontaneous collapse theories, concretely with the simplest and best-known one, the Ghirardi--Rimini--Weber (GRW) theory \cite{GRW86,Bell87} in the versions GRWm \cite{Dio89,BGG95,Gol98} and GRWf \cite{Bell87,Tum04}; for an introduction see, e.g., \cite{AGTZ06,Bell87,Ghi07}. We briefly review GRW theory in Section~\ref{sec:introGRW}. Quantum theories without observers have no need for the orthodox quantum philosophy of ``complementarity.'' In these theories, it is clear what is real. And thus, the possibility arises that certain things that are real cannot be observed.

It is necessary to distinguish between \emph{quantum measurements} and (what we call) \emph{genuine measurements}. Quantum measurements are not literally measurements in the ordinary sense, i.e., they are not experiments for discovering the values of variables that have well-defined values.\footnote{Indeed, it is the content of the ``no-hidden-variable theorems'' such as Gleason's and Kochen--Specker's that one cannot think of the outcomes of quantum measurements as values that were already known to nature before the experiment, and merely made known to us by the experiment.} In contrast, in a quantum theory without observers there are several variables that are supposed to have well-defined values, and we may ask whether and how we could measure these: for example, the wave function $\psi$, the matter density $m(x,t)$ in GRWm, the flashes in GRWf. We call experiments discovering these values genuine measurements.

We show that certain well-defined variables in GRW theories do not permit genuine measurements; that is, that certain facts in GRW worlds cannot be found out by the inhabitants of those worlds. At the same time, there is no way of replacing the GRW theories with simpler and more parsimonious theories by merely denying the factual character of what the inhabitants cannot find out,\footnote{That is, if, in reaction to our result that the flashes cannot be measured accurately, one dropped the flashes from the reality, then one would end up with the GRW wave function without a primitive ontology, which is not a satisfactory physical theory \cite{Mon02,Mon04,AGTZ06,Mau10,All13a,All13b}. Likewise, if, in reaction to our result that the number of collapses in a given time interval $[t_1,t_2]$ cannot be accurately measured, one decided that there is no fact about the number of collapses in $[t_1,t_2]$, then one would have to drop the GRW law of wave function evolution, and thus leave the framework of GRW theories.} much in contrast to an unobservable ether whose existence can very well be denied. 

These limitations to knowledge arise as a consequence of the defining equations of the theories. They are not further, ad hoc postulates; and they do not require carefully contrived, or conspiratorial, initial conditions. Instead, they are dictated by the physical laws. The GRW theories thus exemplify Einstein's dictum that ``it is the theory which decides what can be observed'' \cite{Ein26}. The situation is in a way parallel to that of black holes, where also the physical law itself (in that case, the Einstein equation of general relativity) implies that observers outside the black hole cannot find out what happens inside. We emphasize that the limitations of knowledge in GRW theories do not in any way represent a drawback of these theories.

This paper is organized as follows. In the remainder of Section~1, we provide background on the concept of a limitation to knowledge by describing examples. In Section~\ref{sec:introGRW} we review the definition of the GRW theories. In Section~\ref{sec:main} we introduce our main tool and prove, as a first result, that the function $m(x)$ in GRWm cannot be measured with microscopic accuracy. In Section~\ref{sec:f} we present detailed results about the accuracy with which a collapse can be detected. In Section~\ref{sec:m} we investigate the accuracy with which the function $m(x)$ in GRWm can be measured.

\subsection{Known Examples of Limitations to Knowledge}
 
\begin{enumerate}
\item Some limitations to knowledge are quite familiar: It is impossible to look into the future (e.g., to find out next week's stock prices in less time than a week), to look directly into the past (e.g., to determine systematically who Jack the Ripper was), to look into a spacelike separated region (superluminal signaling), or to look into a black hole.

\item\label{item:exq1} Quantum theory is particularly rich in limitations to knowledge, as exemplified by items \ref{item:exq1}--\ref{item:exqlast} in this list. To begin with, as already mentioned, it is impossible to distinguish empirically between two different ensembles of wave functions with the same density matrix.

\item It is impossible to measure the wave function of an individual system. For example, if Alice chooses a direction in space and prepares a single particle in a pure spin state pointing in this direction then Bob cannot determine this direction by means of whichever experiment on the particle. The best Bob can do is a Stern--Gerlach experiment in (say) the $z$ direction, which yields one bit of information and tells Bob whether Alice's chosen direction is more likely to lie in the upper or lower hemisphere. Bob can do better if Alice prepares $N\gg 1$ disentangled particles, each in the same pure spin state; by means of Stern--Gerlach experiments in different directions and a statistical analysis of the results, he can estimate the direction with arbitrary accuracy (with high probability) for sufficiently large $N$. (It is sometimes suggested that ``protective measurements'' can measure the wave function. However, these experiments involve a mechanism that restores the initial wave function $\psi_0$ of the system after a (weak) interaction with the apparatus, and many repetitions of the procedure. Thus, in effect, the apparatus is provided with many copies of the same wave function $\psi_0$, so that the possibility of determining $\psi_0$ is in agreement with what we said earlier in this paragraph.)

\item A different type of limitation to knowledge comes up in quantum cryptography, specifically in quantum key distribution: The eavesdropper Eve cannot obtain useful information about the key created by Alice and Bob using a quantum key distribution scheme without leaving traces that Alice and Bob can detect.

\item ``Absolute uncertainty'' in Bohmian mechanics \cite{DGZ92}: Given that the (conditional) wave function of a system is $\psi$, it is impossible for an inhabitant of a Bohmian universe to know the system's configuration more accurately than allowed by the $|\psi|^2$ distribution. If the configuration gets measured more accurately, then the conditional wave function becomes narrower. (However, while Bohmian mechanics puts limitations to knowing the configuration, it still allows in principle to measure the configuration with arbitrary accuracy; thus, the access to the configuration is less restricted than the one to the $m$ function in GRWm, which, as we show in Section~\ref{sec:main}, cannot be measured with microscopic accuracy.)

\item It is impossible to measure the velocity of a particle in Bohmian mechanics \cite{DGZ04,DGZ09} (except when information about the wave function is given). That is, there is no machine into which we could insert a particle with arbitrary wave function $\psi$ and which would display correctly the velocity of the particle (i.e., the velocity right before the experiment). Note that the measurement of velocity would amount to (what we called) a \emph{genuine measurement}, not a quantum measurement.

It is possible, in contrast, to build a machine that will correctly measure the velocity if the machine is told what $\psi$ is; for example, the machine could measure the position of the particle (to sufficient accuracy) and then compute the velocity using Bohm's equation of motion. It is also in principle possible to build a machine that correctly displays the velocity \emph{after} the experiment without being given information about the wave function $\psi$ before the experiment; a formal quantum measurement of the momentum observable achieves this. It is also possible to build a machine that correctly displays the velocity for a limited set of $\psi$s (such as approximate momentum eigenstates, i.e., wide packets of plane waves).

\item\label{item:laws} It is impossible to distinguish empirically between certain different versions of Bohmian mechanics \cite{GTTZ05b,DG98}, all of which lead to the appropriate $|\psi|^2$ distribution for macroscopic configurations, or between Bohmian mechanics and Nelson's stochastic mechanics \cite{Nelson,Gol87}, or (presumably) between Bohmian mechanics and some versions of the many-worlds theory \cite{AGTZ11}.

\item It is impossible to distinguish empirically between GRWm and GRWf \cite{AGTZ06,grw3A}.

\item The study of Colin \emph{et al.}\ on superselection rules \cite{CDT05} identified cases in GRWm in which a ``weak'' but no ``strong'' superselection rule holds, which means that for every superposition $\psi$ (of eigenvectors of the superselected operator) there is a mixture $\mu$ (of eigenvectors) that cannot be empirically distinguished from $\psi$, while $\mu$ leads to different histories of the primitive ontology (PO, or local beables) than $\psi$. As a consequence, the difference between the PO arising from $\psi$ and that arising from $\mu$ cannot be detected empirically. The concrete example in GRWm amounts more or less to the fact that the $m$ function cannot be measured with microscopic accuracy.

\item\label{item:exqlast} Whether the Heisenberg uncertainty relation is or is not an instance of a limitation to knowledge, depends on the precise version of quantum mechanics used. 

In Bohmian mechanics, it is: Certain standard experiments realizing a ``quantum measurement of the momentum observable'' (such as letting a particle move freely for a long time $t$, then measuring its position $\vQ(t)$ to sufficient accuracy, and multiplying by $m/t$) actually measure (mass times) the long-time average of the Bohmian velocity, $\vu=\lim_{t\to\infty} \vQ(t)/t$, which is a deterministic function of the initial position $\vQ(0)$ and the initial wave function $\psi_0$, $\vu=\vu(\vQ(0),\psi_0)$. The Heisenberg uncertainty relation implies that for a Bohmian particle with position $\vQ$ and wave function $\psi$, even if $\psi$ is known, the values of $\vQ$ and $\vu(\vQ,\psi)$ cannot both be known with arbitrary accuracy, although both quantities have precise values in reality.\footnote{It is possible, however, to measure $\vQ$ (to sufficient accuracy) and then calculate $\vu(\vQ,\psi)$ if $\psi$ is known. Since this measurement will change the wave function, it is afterwards still not known what the present value of $\vu$ is; we only obtain information about what $\vu(\vQ,\psi)$ was before the measurement, not about what it is now after the measurement.}

In collapse theories such as GRWm and GRWf, in contrast, there is no precise value of either the position or the momentum observable before $\E$, if $\E$ is an experiment that can be regarded as a ``quantum measurement of position or momentum.'' In particular, $\E$ is not a measurement in the literal sense. Rather, the outcome of such an experiment is a random value that is only generated in the course of the experiment. Since, as long as no such experiment is carried out, there is no fact about the value of the position or the momentum observable, there is nothing to be ignorant of. As a consequence, the Heisenberg uncertainty relation does not constitute a limitation to knowledge in GRW theories.

\item Chaos leads to practical limitations: If the behavior of a (classical) dynamical system depends sensitively on the initial conditions, and if our knowledge of the initial conditions has limited accuracy, then we may be unable to predict the behavior although the system is deterministic in principle. This fact can be regarded as a limitation to knowledge, too, but there is a fundamental difference to the limitations of knowledge discussed in this paper: There is no reason in sight for why this limitation should be unsurmountable. If we are willing to make a bigger effort when measuring the initial conditions so as to obtain higher accuracy, and if we are willing to make a bigger effort in the computation of predictions, then we may be able to predict the behavior of the system for a longer time interval.
\end{enumerate}

\subsection{Remarks}

\begin{enumerate}
\item There are parallels between the limitations to knowledge discussed in this paper and G\"odel's incompleteness theorem \cite{G31}. While G\"odel's theorem concerns a mathematical statement that is formally undecidable, our results concern physical statements that are empirically undecidable. The distinction between true and provable in mathematics resembles the distinction between real and observable in physics. Mathematical platonism, which can be described as the view that a mathematical statement can be true even if it is not provable, is parallel to realism, which can be described as the view that something can be a physical fact even if it is not observable; mathematical formalism, the opposite view, is parallel to physical positivism. A basic difference between G\"odel's theorem and our results is that we \emph{can} know the truth value of G\"odel's undecidable statement: it is true. That is, while the truth value is unknowable to the formal system, it is knowable to us. Thus, the situation of G\"odel's undecidable statement is more analogous to that of the physical laws (as discussed in item \ref{item:laws} above) than to that of measuring, say, the number of collapses. Namely, while we often have no more information about the number of collapses than its a priori probability distribution, we may be able to guess the physical laws, and thus prefer one among several empirically equivalent theories, on the basis of the simplicity, naturalness, elegance, and plausibility of the theory, even though our intuition will not provide as reliable and clear a decision as about the truth value of G\"odel's undecidable statement. Finally, we note that there may be, besides G\"odel's statement, other examples of formally undecidable statements that leave us completely and permanently in the dark as to what their truth values are.

\item There are also parallels between the limitations to knowledge discussed in this paper and Carnot's theory of heat engines: Not all of the energy contained in a system can be extracted in a useful form (i.e., as work); not all of the information contained in a system can be extracted in a useful form (i.e., as human knowledge). By the latter statement we mean that not all of the facts true of a system can be found out empirically.

\end{enumerate}

\section{Brief Review of GRW Theories}
\label{sec:introGRW}

\subsection{The GRW Process}
\label{sec:grwprocess}

In both the GRWm and the GRWf theory the evolution of the wave function follows, instead of the Schr\"odinger equation, a stochastic jump process in Hilbert space, called the GRW process. Consider a quantum system of (what would normally be called) $N$ ``particles,'' described by a wave function $\psi = \psi(q_1, \ldots,q_N)$, $q_i\in \RRR^3$, $i=1,\dots, N$.  The GRW process behaves as if an ``observer'' outside the universe made unsharp ``quantum measurements'' of the position observable of a randomly selected particle at random times $T_1,T_2,\ldots$ that occur with constant rate $N\lambda$, where $\lambda$ is a new constant of nature of order of $10^{-16}\, \text{s}^{-1}$, called the collapse rate per particle. The wave function ``collapses'' at every time $T=T_k$, i.e., it changes discontinuously and randomly as follows. The post-collapse wave function $\psi_{T+}=\lim_{t\searrow T}\psi_t$ is obtained from the pre-collapse wave function $\psi_{T-}=\lim_{t\nearrow T} \psi_t$ by multiplication by a Gaussian function,
\be\label{collapse}
\psi_{T+}(q_1,\ldots,q_N) = \frac{1}{\N} g_\sigma(q_I-X)^{1/2}\, \psi_{T-}(q_1,\ldots,q_N) \,,
\ee
where
\be\label{Gaussian}
g_\sigma(x) = \frac{1}{(2\pi\sigma^2)^{3/2}} e^{-\frac{x^2}{2\sigma^2}}
\ee
is the 3-dimensional Gaussian function of width $\sigma$, $I$ is chosen randomly from $1,\ldots,N$, and
\be\label{scrNdef}
\N=\N(X)=\biggl(\int_{\RRR^{3N}} dq_1\cdots dq_N\, g_\sigma(q_I-X)\, |\psi_{T-}(q_1,\ldots,q_N)|^2\biggr)^{1/2}
\ee
is a normalization factor. The width $\sigma$ is another new constant of nature of order of $10^{-7}\, \text{m}$, while the center $X=X_k$ is chosen randomly with probability density
\be
\rho(x)=\N(x)^2\,.
\ee
We will refer to $(X_k,T_k)$ as the space-time location of the collapse.

Between the collapses, the wave function evolves according to the Schr\"odinger equation corresponding to the standard Hamiltonian $H$ governing the system, e.g., given, for $N$ spinless particles, by
\begin{equation}\label{eq:H}
  H=-\sum_{i=1}^N\frac{\hbar^2}{2m_i}\nabla^2_{q_i}+V,
\end{equation}
where $m_i$, $i=1, \ldots, N$, are the masses of the particles, and $V$ is the potential energy function of the system. Due to the stochastic evolution, the wave function $\psi_t$ at time $t$ is random.

This completes our description of the GRW law for the evolution of the wave function. According to the GRW theories, the wave function $\psi$ of the universe evolves according to this stochastic law, starting from the initial time (say, the big bang). As a consequence \cite{BG03,grw3A}, a subsystem of the universe (comprising $M<N$ ``particles'') will have a wave function $\varphi$ of its own that evolves according to the appropriate $M$-particle version of the GRW process during the time interval $[t_1,t_2]$, provided that $\psi(t_1)=\varphi(t_1)\otimes\chi(t_1)$ and that the system is isolated from its environment during that interval.

\bigskip

We now turn to the primitive ontology (PO), that is, the part of the ontology (i.e., of what exists, according to the theory) that represents matter in space and time (and of which macroscopic objects consist), according to the theory. Without such further ontology, the GRW theory would not be satisfactory as a fundamental physical theory \cite{Mon02,Mon04,AGTZ06,Mau10,All13a,All13b}. In the subsections below we present two versions of the GRW theory, based on two different choices of the PO, namely the \emph{matter density ontology} (GRWm in Section~\ref{sec:GRWm}) and the \emph{flash ontology} (GRWf in Section~\ref{sec:GRWf}).

\subsection{GRWm}
\label{sec:GRWm}

GRWm postulates that, at every time $t$, matter is continuously distributed in space with density function $m(x,t)$ for every location $x\in \RRR^3$, given by
\begin{align}
 m(x,t) 
 &= \sum_{i=1}^N m_i \int\limits_{\RRR^{3N}}  dq_1 \cdots dq_N \, \delta^3(q_i-x) \,  
 \bigl|\psi_t(q_1, \ldots, q_N)\bigr|^2 \label{mdef}\\
&= \sum_{i=1}^N m_i \int\limits_{\RRR^{3(N-1)}} \!\!\!\! dq_1\cdots dq_{i-1} \, dq_{i+1}\cdots dq_N\,
\bigl| \psi_t(q_1,\ldots,q_{i-1},x,q_{i+1},\ldots,q_N) \bigr|^2\,.
\end{align}
In words, one starts with the $|\psi|^2$--distribution in configuration space $\RRR^{3N}$, then obtains the marginal distribution of the $i$-th degree of freedom $q_i\in \RRR^3$
by integrating out all other variables $q_j$, $j \neq i$, multiplies by the mass associated with $q_i$, and sums over $i$. Alternatively, \eqref{mdef} can be rewritten as
\begin{equation}\label{mdef2}
  m(x,t) = \scp{\psi_t}{\widehat{M}(x) |\psi_t}
\end{equation}
with
\be\label{hatMdef}
\widehat{M}(x) = \sum_{i} m_i \, \delta^3(\widehat{Q}_i - x)
\ee
the mass density operator, defined in terms of the position operators $\widehat{Q}_i \psi(q_1,\ldots,q_N) = q_i\,\psi(q_1,\ldots,q_N)$.

\subsection{GRWf}
\label{sec:GRWf}

According to  GRWf, the PO is given by ``events'' in space-time called flashes, mathematically described by points in space-time. What this means is that in GRWf matter is neither made of particles following world lines, nor of a continuous distribution of matter such as in GRWm, but rather of discrete points in space-time, in fact finitely many points in every bounded space-time region.

In the GRWf theory, the space-time locations of the flashes can be read off from the history of the wave function: every flash corresponds to one of the spontaneous collapses of the wave function, and its space-time location is just the space-time location of that collapse. The flashes form the set
\be
  F=\{(X_{1},T_{1}), \ldots, (X_{k},T_{k}), \ldots\}
\ee
(with $T_1<T_2<\ldots$). Alternatively, we may postulate that flashes can be of $N$ different types (``colors''), corresponding to the mathematical description
\be
F=\{(X_1,T_1,I_1),\ldots, (X_k,T_k,I_k),\ldots\}\,,
\ee
with $I_k$ the number of the particle affected by the $k$-th collapse.

Note that if the number $N$ of degrees of freedom in the wave function is large, as in the case of a macroscopic object, the number of flashes is also large (if $\lambda=10^{-15}$ s$^{-1}$ and $N=10^{23}$, we obtain $10^{8}$ flashes per second). Therefore, for a reasonable choice of the parameters of the GRWf theory, a cubic centimeter of
solid matter contains more than $10^8$ flashes per second. That is to say that large numbers of flashes can form macroscopic shapes, such as tables and chairs. That is how we find an image of our world in GRWf.

\bigskip

We should remark that the word ``particle'' can be misleading. According to GRWf, there are no particles in the world, just flashes and a wave function. According to GRWm, there are no particles, just continuously distributed matter and a wave function. The word ``particle'' should thus not be taken literally (just like, e.g., the word ``sunrise''); we use it only because it is common terminology in quantum mechanics.

\section{First Examples of Limitations to Knowledge in GRW Theories}
\label{sec:main}

An important tool for the analysis of limitations to knowledge is the \emph{main theorem about POVMs}, which says that \emph{for every experiment $\E$ on a system ``$\sys$'' there is a POVM (positive-operator-valued measure\footnote{A \emph{POVM} is a family of positive operators $E_z$ such that $\sum_z E_z=I$, the identity operator. It is also known as a \emph{generalized observable} and in fact generalizes the notion of a quantum observable represented by a self-adjoint operator $A$, which applies to an \emph{ideal quantum measurement}. For an ideal quantum measurement, the values $z$ are the eigenvalues of $A$ and $E_z$ is the projection to the eigenspace.}) $E$ on the value space of $\E$ acting on the system's Hilbert space $\Hilbert_\sys$ such that if $\sys$ has wave function $\psi$ and $\E$ is carried out then the outcome $Z$ has probability distribution}
\be\label{mainthmPOVM}
\PPP(Z=z)=\scp{\psi}{E_z|\psi}\,.
\ee
This theorem has been proven for Bohmian mechanics \cite{DGZ04}, GRWf \cite{Tum07,grw3A}, and GRWm \cite{grw3A}; a similar result for GRWm can be found in \cite{BGS07}. In orthodox quantum mechanics, the theorem is true as well, taking for granted that, after $\E$, a quantum measurement of the position observable of the pointer of $\E$'s apparatus will yield the result of $\E$.

From the main theorem about POVMs we can deduce a first limitation to knowledge \cite{grw3A}: that it is impossible to measure the matter density $m(x,t)$ in GRWm. More precisely, it is impossible to build a machine that will, when fed with a system with any wave function $\psi$, determine $m(x)$. This is because the outcome $Z$ of any experiment in a GRW world has a probability distribution $\PPP(Z=z)$ whose dependence on the wave function is quadratic, $\scp{\psi}{E(z)|\psi}$, while the $m(x)$ function (or, in fact, any functional of the wave function) is deterministic in $\psi$, that is, its probability distribution is a Dirac delta function and not quadratic.\footnote{The fact that the deterministic value of $m(x)$ is given by a quadratic expression in $\psi$, viz.~\eqref{mdef}, should not be confused with the condition that the probability distribution of $m(x)$ depend on $\psi$ in a quadratic way.} This result notwithstanding, it \emph{is} possible to measure $m(x)$ with limited accuracy, that is, to measure a macroscopic, coarse-grained version of $m(x)$; this we will study in Section~\ref{sec:m}.

The same type of argument shows \cite{grw3A} that it is impossible to measure the wave function $\psi_t$ of a system in either GRWm or GRWf (or in Bohmian mechanics, many-worlds, or orthodox quantum mechanics, for that matter).
 
Let us compare the situation in GRW theories to that of Bohmian mechanics. As mentioned above, the velocity of a given Bohmian particle is not measurable. On the other hand, there is no limitation in principle in Bohmian mechanics to measuring the position of a particle to arbitrary accuracy, except that doing so will alter the particle's (conditional) wave function, and thus its future trajectory. Here we encounter a basic difference between Bohmian mechanics and GRWm: the configuration of the primitive ontology can be measured in Bohmian mechanics but not in GRWm. (In Bohmian mechanics, the configuration of the primitive ontology corresponds to the positions of all particles, while in GRWm it corresponds to the $m(x,t)$ function for all $x\in\RRR^3$.) In GRWf, for comparison, there is nothing like a configuration of the primitive ontology \emph{at time $t$}, of which we could ask whether it can be measured. There is only a space-time \emph{history} of the primitive ontology, which we may wish to measure. Bohmian mechanics is an example of a world in which the history of a system cannot be measured without disturbing its course, and indeed disturbing it all the more drastically the more accurately we try to measure it. This suggests already that also in GRWf, measuring the pattern of flashes will entail disturbing it---and thus finding a pattern of flashes that is different from what would have occurred naturally (i.e., without intervention), so that this experiment could not be regarded as a genuine measurement of the pattern of flashes.

\section{Measurements of Flashes in GRWf, or of Collapses in GRWm}
\label{sec:f}

Even if we accept disturbances, the following heuristic reasoning suggests that individual flashes cannot be detected. Suppose we had an apparatus capable of detecting flashes in a system. Think of the wave function of system and apparatus together as a function on configuration space $\RRR^{3N}$, and think of configurations as one would in Bohmian mechanics. Let $R_0$ be the set of those configurations $q$ such that in a Bohmian world in configuration $q$ the apparatus display reads ``no flash detected so far,'' and let $R_1$ be the set of those configurations in which the display reads ``one flash detected so far.'' Then $R_1$ is disjoint from $R_0$. Recall that a flash in the system leads to a change in the wave function of the form
\be\label{collapse1}
\psi\to\psi'=\frac{1}{\N} g_\sigma(q_i-x)^{1/2} \,\psi \,.
\ee
But such a change does not push the wave function from $R_0$ to $R_1$. That is, if $\psi$, as a function on $\RRR^{3N}$, is concentrated in the region $R_0$, then $\psi'$ as given by \eqref{collapse1} will not be concentrated in $R_1$; instead, it is still concentrated in (some subset of) $R_0$.

While the macroscopic equivalence class\footnote{Consider two patterns of flashes in space-time. We call them macroscopically equivalent if they look alike on the macroscopic level. While the notion of macroscopic equivalence is not precisely defined, it is roughly defined.} of the pattern of flashes is measurable, we are led to suspecting that the microscopic details of the pattern are not. In what sense exactly that is or is not the case will be discussed in this section.

Note first that the main theorem about POVMs does not directly exclude measurements of $F$, the pattern of flashes, in the way it directly excluded measurements of $m(x)$ or $\psi$. After all, the argument was that the probability distribution of $m(x)$ or $\psi$ does not depend quadratically on $\psi$. The probability distribution of $F$, in contrast, \emph{does} depend on $\psi$ in a quadratic way: There is a continuous POVM $G(f)$ on the space of all flash histories $f$ such that the probability density of $F$ is given by
\be
\PPP(F=f) = \scp{\psi}{G(f)|\psi}\,,
\ee
with $\psi$ the initial wave function \cite{Tum07,grw3A}. Of course, this fact does not imply that $F$ can be measured---and we are claiming that it cannot.

To approach the question whether one can detect an individual flash (or an individual collapse in GRWm), we begin with a simple example. 

\subsection{An Example in which $\psi$ is Known}

Suppose $\psi$ is the wave function of a single particle and a superposition of two wave packets with disjoint supports in space,
\be\label{superposition}
\psi = \tfrac{1}{\sqrt{2}} |\text{here}\rangle + \tfrac{1}{\sqrt{2}} |\text{there}\rangle\,,
\ee
as may result from a double-slit setup. (The reasoning that is to follow will also apply to a molecule or small solid body with $|\text{here}\rangle$ and $|\text{there}\rangle$ differing by a shift in the center-of-mass coordinate that makes their supports disjoint.)
Suppose, for simplicity, that the Hamiltonian of the system vanishes, so that the Schr\"odinger time evolution is trivial. We ask whether any flash at all occurs during the time interval $[t_1,t_2]$. We are interested in the case in which the probability of a flash is neither close to 1 nor close to 0, a case that can be arranged by suitable choice of the duration $t_2-t_1$. (For a single particle, this choice might mean the duration is millions of years; we might either consider this case as a theoretical exercise, or consider instead the center-of-mass of a many-particle system to reduce the duration. To obtain a reasonable duration, we may want to consider the case that the number of particles is big (say, $>10^{10}$) but not too big (say, $<10^{20}$); anyway, that makes no difference to the theoretical analysis.) 

For simplicity, we ignore the possibility of multiple collapses and assume that a collapse occurs with probability $p$, and no collapse with probability $1-p$. Let us further assume that the two packets $|\text{here}\rangle$ and $|\text{there}\rangle$ have width less than $\sigma$ but separation greater than $\sigma$, so that after collapse the wave function is either approximately $|\text{here}\rangle$ or approximately $|\text{there}\rangle$. For simplicity, let us assume that after collapse the wave function is either exactly $|\text{here}\rangle$ or exactly $|\text{there}\rangle$, each with probability $1/2$. Thus the final wave function $\psi'$ is distributed according to
\begin{align}
\PPP(\psi'=\psi)&=1-p\,,\nonumber\\
\PPP(\psi'=|\text{here}\rangle)&=p/2\,\label{psiprimedistr}\\
\PPP(\psi'=|\text{there}\rangle)&=p/2\,.\nonumber
\end{align}
Let $C\in\{0,1\}$ be the random number of collapses.

We ask whether and how well an experiment beginning at time $t_2$ can decide whether a collapse has occurred. This amounts essentially to distinguishing between the three vectors $|\text{here}\rangle$, $|\text{there}\rangle$, and $\psi$. While $|\text{here}\rangle$ and $|\text{there}\rangle$ are orthogonal, $\psi$ is not orthogonal to either. As is well known, it is not possible to distinguish reliably between non-orthogonal vectors. (Our question can also be regarded as a special case of the question how well an experiment can distinguish between two density matrices $\rho_1$, $\rho_2$; see Section~\ref{sec:rho1rho2}.)

The following experiment $\E_1$ provides probabilistic information about $C$: carry out a ``quantum measurement of the observable'' $E_1$ given by the projection to the 1-dimensional subspace orthogonal to that spanned by \eqref{superposition},
\be\label{E1def}
E_1=I-\pr{\psi}
\ee
with $I$ the identity operator. 
If the result $Z$ was 1, then it can be concluded that a collapse has occurred, $C=1$. (Because if no collapse has occurred, then $\psi'=\psi$ and $Z=0$ with probability 1.) If the result $Z$ was 0, nothing can be concluded with certainty (since also $|\text{here}\rangle$ and $|\text{there}\rangle$ lead to a probability of $1/2$ for the outcome to be 0). However, in this case the (Bayesian) conditional probability that a collapse has occurred is less than $p$ (and thus $Z$ is informative about $C$):
\begin{align}
  \PPP(C=1|Z=0)&=\frac{\PPP(C=1,Z=0)}{\PPP(C=1,Z=0)
  +\PPP(C=0,Z=0)}\\[4mm]
  &=\frac{\PPP(Z=0|C=1)\,\PPP(C=1)}{\PPP(Z=0|C=1)\,
  \PPP(C=1)+\PPP(Z=0|C=0)\,\PPP(C=0)}\\[3mm]
  &= \frac{\tfrac{1}{2}p}{\tfrac{1}{2}p + 1\cdot(1-p)} = \frac{p}{2-p}<p\,.
\end{align}
Thus, in every case the experiment can retrodict $C$ with greater precision than it could have been predicted a priori (i.e., than attributing probability $p$ to a collapse and $1-p$ to no collapse). 

To quantify the usefulness of the experiment, we define the \emph{reliability} $R(\E)$ of a yes-no experiment (or 1-0 experiment) $\E$ as the probability that it correctly retrodicts whether a collapse has occurred,
\begin{align}
R(\E) &= \PPP(Z=0,C=0)+\PPP(Z=1,C=1)\\[2mm]
&=\PPP(Z=0|C=0)\, \PPP(C=0) + \PPP(Z=1|C=1)\, \PPP(C=1)\,.
\end{align}
For the particular experiment $\E_1$ just described, we find that $\PPP(Z=0|C=0)=1$, $\PPP(C=0)=1-p$, $\PPP(Z=1|C=1)=\frac12$, and $\PPP(C=1)=p$, so
\be\label{RE1}
R(\E_1)=1-\frac{p}{2} \,.
\ee
(See Proposition~\ref{prop:reliability} in Appendix~\ref{sec:proofs} for a more general result.) 
The fact that this quantity is less than 1 means that this experiment cannot decide with certainty whether a collapse has occured.

\begin{prop}\label{prop:E1} \cite{CT12a}
For the initial wave function \eqref{superposition} and $0\leq p \leq 2/3$, no experiment at time $t_2$ can retrodict $C$ with greater reliability than the quantum measurement of $E_1=I-\pr{\psi}$:
\be\label{boundE1}
\forall \E \: \forall p\in [0,\tfrac{2}{3}]: \quad R(\E)\leq 1-\frac{p}{2}\,.
\ee
In particular, for $p\neq 0$ it is impossible to determine with reliability 1 whether a collapse has occurred or not.
\end{prop}

The proof \cite{CT12a} of this proposition relies, of course, on the main theorem about POVMs, which associates with every yes-no experiment (acting on a system with wave functions that are superpositions of $|\text{here}\rangle$ and $|\text{there}\rangle$) a positive semi-definite $2\times 2$ matrix $E$ (namely, $E=E_{\yes}$, while $E_{\no}=I-E_{\yes}$). The proof shows that there is no such matrix $E$ for which $R$ exceeds $1-p/2$. We note that $R(\E)$ depends on $\E$ only through $E$, that is, two experiments with the same POVM have the same reliability (see Proposition~\ref{prop:reliability} in Appendix~\ref{sec:proofs}).

Proposition~\ref{prop:E1} expresses a limitation to knowledge: Although we can empirically gain \emph{some} information about the value of $C$, we cannot obtain full information, i.e., we cannot measure $C$ with certainty. In fact, not even close to certainty: While the maximal reliability $1-p/2$ is close to 1 if $p$ is small, in this case it is easy to guess correctly without any experiment whether a collapse occurred: no. In other words, the reliability $1-p/2$ can be put into perspective by comparing it to that of \emph{blind guessing}, i.e., of the experiment $\E_\emptyset$ that does not even interact with the system but always answers ``no'' if $p\leq 1/2$ and always ``yes'' if $p>1/2$. (We think of $p$ as known.) This experiment (which corresponds to $E=0$ if $p\leq 1/2$ and to $E=I$ if $p>1/2$) has reliability
\be\label{RE0}
R(\E_\emptyset)=\max\{p,1-p\}\,.
\ee
For $0\leq p \leq 1/2$, the optimal reliability is right in the middle between the trivial achievement $R=1-p$ and the desired achievement $R=1$; for $1/2\leq p \leq 2/3$, it is even closer to the trivial achievement $R=p$. For $p>2/3$, the situation is even worse:

\begin{prop}\label{prop:E0} \cite{CT12a}
For $\psi$ as in \eqref{superposition} and $2/3\leq p \leq 1$, no experiment at time $t_2$ can retrodict $C$ with greater reliability than blind guessing:
\be\label{boundE0}
\forall \E \: \forall p\in [\tfrac{2}{3},1]: \quad R(\E)\leq p\,.
\ee
\end{prop}

See Figure~\ref{figone}.

\begin{figure}[h]
\begin{center}
\includegraphics[width=.5 \textwidth]{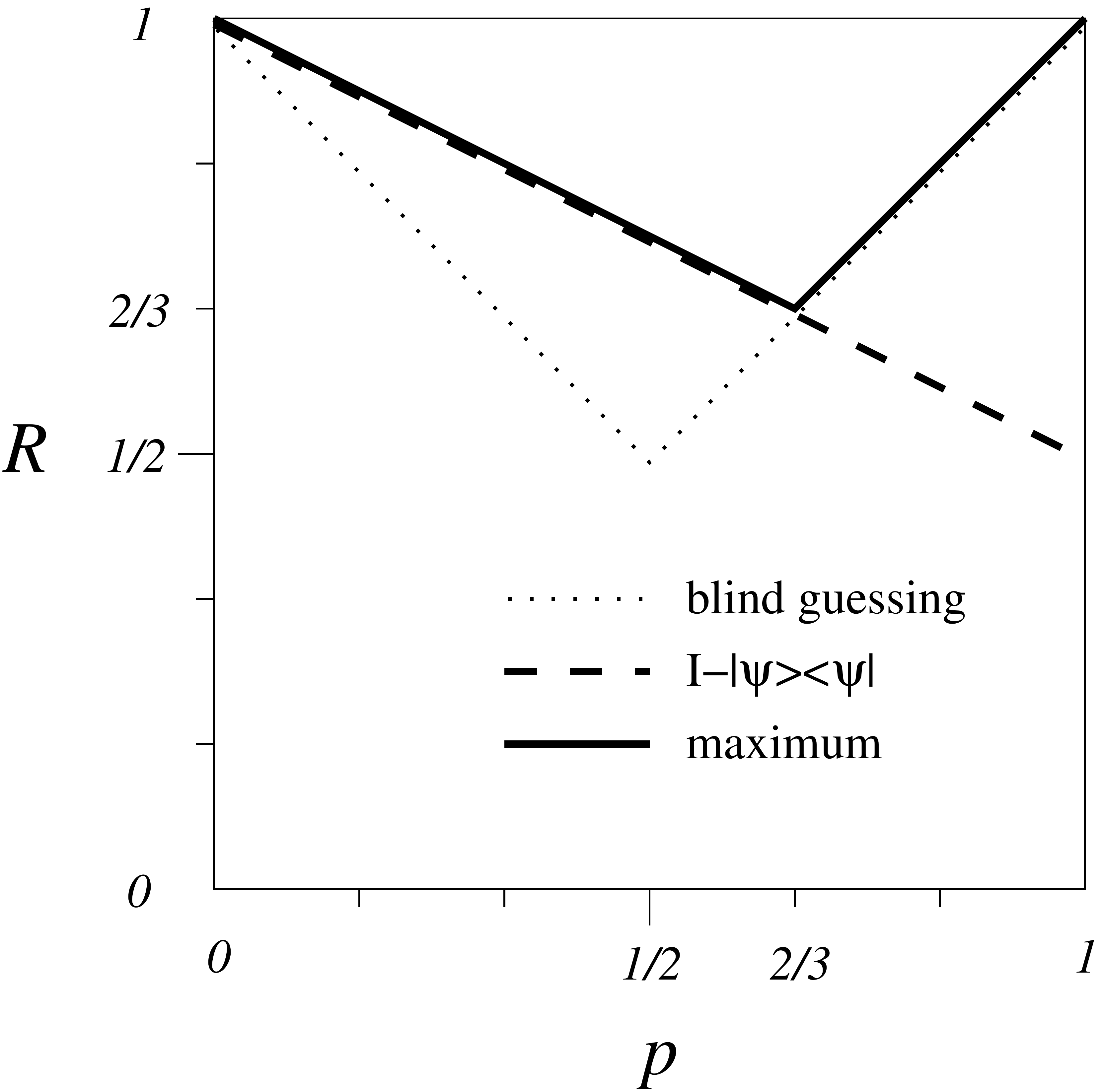}
\end{center}
\caption{
Reliability of experiments for detecting a collapse of the wave function \eqref{superposition} to one of the two contributions, as a function of $p$, the probability of collapse. Graphs are shown for blind guessing, for the quantum measurement of $E_1=I-\pr{\psi}$, and for the maximal value of any experiment.}
\label{figone}
\end{figure}

One can understand easily why the optimal strategy changes at $p=2/3$. Suppose we carry out $\E_1$ (corresponding to $I-\pr{\psi}$) and obtain outcome $Z$; as we saw above, if $Z=1$ then a collapse must have occurred, and if $Z=0$ then the probability that a collapse occurred is $p/(2-p)$; this value is $<1/2$ for $p<2/3$ and $>1/2$ for $p>2/3$. Thus, for $p<2/3$ the best retrodiction in case $Z=0$ is that no collapse occurred, whereas for $p>2/3$ it is better to always answer ``yes'' in both cases, $Z=1$ and $Z=0$.

\subsection{Other Choices of $\psi$}
\label{sec:otherchoices}

The fact that knowledge is limited does not depend on whether the superposition \eqref{superposition} involved two or more contributions, nor on whether the coefficients of the contributions were equal. Consider an arbitrary initial wave function $\psi$ from the unit sphere $\SSS(\Hilbert)$ of an $n$-dimensional Hilbert space $\Hilbert$ and an orthonormal basis $B=\{b_1,\ldots,b_n\}$ of $\Hilbert$. Suppose that collapse occurs with probability $p$ and projects $\psi$ to one of the basis vectors with the quantum probability; that is, the final state vector $\psi'$ has distribution
\begin{align}
\PPP\bigl(\psi'=\tfrac{\scp{b_k}{\psi}}{|\scp{b_k}{\psi}|}b_k\bigr) &= |\scp{b_k}{\psi}|^2\,p\,,\label{psiprimedef1}\\
\PPP(\psi'=\psi)&=1-p\,.\label{psiprimedef2}
\end{align}
Then the reliability, which depends on $\psi$ and on the experiment $\E$, is bounded by \cite{CT12a}
\be\label{boundRpn}
R_\psi(\E)\leq 1-\frac{p}{n} \quad \forall\psi\in\SSS(\Hilbert) \: \forall \E \: \forall p\in[0,\tfrac{n}{n+1}]\,.
\ee
Equality holds for $\psi=\sum_k n^{-1/2} b_k$ and $\E$ the quantum measurement of $E=I-\pr{\psi}$. For other $\psi$, the optimal reliability $R_\psi= \max_\E R_\psi(\E)$ is even less than $1-p/n$; for $\psi=2^{-1/2}b_1+2^{-1/2}b_2$ it is still $R_\psi=1-p/2$. Curiously, for generic $\psi$ the optimal experiment is different from the quantum measurement of $I-\pr{\psi}$; it is still of the form $I-\pr{\tilde\psi}$ but with $\tilde\psi\neq \psi$; we will give more detail around \eqref{Eopt} below. 

For $p>n/(n+1)$, no experiment is more reliable than blind guessing (for which \eqref{RE0} is still valid),
\be\label{boundRpnLarge}
R_\psi(\E)\leq p \quad \forall\psi\in\SSS(\Hilbert)\: \forall \E \: \forall p\in[\tfrac{n}{n+1},1]\,.
\ee

\subsection{Experiments Beginning Before $t_2$}

So far we have considered only experiments beginning at $t_2$. One might think that an apparatus might do better that interacts with the system during $[t_1,t_2]$, for example because it might seem easier to detect a flash when it happens than at a later time. However, as mentioned already at the end of Section~\ref{sec:main}, when trying to \emph{measure} the flashes in a system during $[t_1,t_2]$ we do not want to \emph{disturb} them; that is, we want them to occur in the same pattern as they would have occurred without intervention. Of course, in a stochastic theory it is not meaningful to ask after an intervention during $[t_1,t_2]$ whether the pattern of flashes in $[t_1,t_2]$ would have been the same had no intervention occurred. It is meaningful, however, to ask whether the \emph{distribution} of the flashes would have been the same had no intervention occurred. We now show that the distribution will in fact be changed by any informative intervention and conclude that an experiment interacting with the system before $t_2$ will necessarily disturb the flashes during $[t_1,t_2]$.

The key to showing that the probability distribution of the flashes after the intervention is different from what it would have been without intervention is that the wave function of the system gets changed by the intervention. Indeed, imagine an experiment $\E_\mathrm{nd}$ (nd stands for ``non-disturbing'') that, when applied to a system with a wave function $\psi'$ as in \eqref{psiprimedistr}, will return the system undisturbed, with the same wave function $\psi'$. Then $\E_\mathrm{nd}$ will not reveal whether a collapse has occurred or not (i.e., whether $\psi'=\psi$ or not), and in fact will yield no information at all about this question; that is, if $Z$ is the outcome of $\E_\mathrm{nd}$ and $C$ the number of collapses (i.e., $C=0$ if $\psi'=\psi$ and $C=1$ otherwise), then the conditional distribution of $Z$ given $C$ does not depend on $C$,
\be
\PPP(Z=z|C=0)=\PPP(Z=z|C=1)
\ee
for all $z$.\footnote{A slightly stronger statement is true: Suppose that $\E_\mathrm{nd}$, when applied to a system with a wave function $\psi'$ that is either $|\text{here}\rangle$ or $|\text{there}\rangle$ or $\psi$ as in \eqref{superposition}, will return the system undisturbed. Then $\E_\mathrm{nd}$ will not reveal which state $\psi'$ is, and in fact will yield no information at all about $\psi'$; that is, the conditional distribution of $Z$, given $\psi'$, does not depend on what $\psi'$ is,
$\PPP\bigl(Z=z\big|\psi'=|\text{here}\rangle\bigr)
=\PPP\bigl(Z=z\big|\psi'=|\text{there}\rangle\bigr) 
= \PPP\bigl( Z=z \big| \psi'=\psi\bigr)$. This follows from standard quantum measurement theory.}

Indeed, consider applying $\E_\mathrm{nd}$ at time $t_2$ repeatedly, $N$ times, to a system with wave function $\psi'$ as in \eqref{psiprimedistr} without giving it the possibility to collapse in between. For $c\in\{0,1\}$, let $P_{z,c}=\PPP(Z=z|C=c)$ with $Z$ the outcome of a single run of $\E_\mathrm{nd}$ on a system in state $\psi'$ as in \eqref{psiprimedistr}. Since the outcomes of the repeated runs of $\E_\mathrm{nd}$ are stochastically independent, the number of $z$-occurrences will, conditionally on $C=c$, have a binomial distribution with parameters $N$ and $P_{z,c}$, so that $P_{z,C}$ can be read off from the number of $z$-occurrences with reliability arbitrarily close to 1, provided that $N$ is sufficiently large. If, for any $z$, $P_{z,1}\neq P_{z,0}$ then the value of $C$ could be read off from that of $P_{z,C}$, so $C$ could be determined as reliably as desired, which contradicts the bound of Proposition~\ref{prop:E1}. Thus, $P_{z,1}=P_{z,0}$.

\subsection{If $\psi$ is Random}
\label{sec:randompsi}

In the previous subsections, we assumed that the initial (i.e., pre-collapse) wave function $\psi$ is known. Now assume that $\psi$ is not known, but random with known probability distribution $\mu$ over the unit sphere $\SSS(\Hilbert)$. So we have limited information about $\psi$. Then the reliability of a yes-no experiment $\E$, i.e., the probability that the experiment correctly answers whether a collapse has occurred, is easily found to be
\be
R_\mu(\E) = \EEE_\mu \, R_\psi(\E) = \int\limits_{\SSS(\Hilbert)} \mu(d\psi) \, R_\psi(\E)\,.
\ee
From \eqref{boundRpn} and \eqref{boundRpnLarge}, we immediately have that for all $\mu$ and all $\E$,
\be
R_\mu(\E) \leq 
 \begin{cases} 
   1-\frac{p}{n} & \text{if }0\leq p \leq \tfrac{n}{n+1}\\ 
   p & \text{if }\tfrac{n}{n+1} \leq p \leq 1\,.
 \end{cases}
\ee
We also note that $R_\mu(\E)$ depends on $\mu$ only through $\rho_\mu$; that is, two distributions with the same density matrix lead to the same reliability for each experiment, $R_\mu(\E)=R_\rho(\E)$ for $\rho=\rho_\mu$ (see Proposition~\ref{prop:reliability} in Appendix~\ref{sec:proofs}).\footnote{In fact, the same value $R_\rho(\E)$ of reliability applies also when the density matrix $\rho$ arises not from a mixture with distribution $\mu$ but as the reduced density matrix when the system $S$ under consideration is entangled with another system $T$, i.e., $\rho = \tr_T \pr{\psi}$ with $\psi$ the pure state of $S$ and $T$ together \cite{CT12a}.}

An extreme case is the uniform distribution $u$ over $\SSS(\Hilbert)$, or any distribution $\mu$ with density matrix $\rho_\mu=\rho_u=\frac{1}{n}I$. For this case we obtain a stronger limitation:

\begin{prop}\label{prop:uniform}
Let $u$ be the uniform distribution on $\SSS(\Hilbert)$. Then 
\be
R_u(\E)\leq \max\{p,1-p\}
\ee
for every experiment $\E$ at time $t_2$. That is, no experiment is more reliable than blind guessing. If $p=\frac{1}{2}$ then, in fact, $R_u(\E)=1/2$ for every experiment $\E$ at time $t_2$.
\end{prop}

We give the proof in Appendix~\ref{sec:proofs}.

Proposition~\ref{prop:uniform} expresses a severe limitation to knowledge: If $\psi$ is uniformly distributed then the reliability of any experiment is no greater than that of blind guessing. Here is an even stronger statement: If $\psi$ is uniformly distributed then no experiment can convey any information at all about whether or not a collapse has occurred. More precisely, no experiment on the system can produce an outcome $Z$ such that the conditional distribution $\PPP(C|Z)$ would be any different from the a priori distribution $(p, 1-p)$:
 
\begin{prop}\label{prop:noinf}
Consider the collapse process as in \eqref{psiprimedef1}--\eqref{psiprimedef2} and an arbitrary experiment $\E$ at time $t_2$, possibly with more than two possible outcomes. Let $\psi$ be uniformly distributed on $\SSS(\Hilbert)$. Then $\PPP(C=1|Z)=p$ and $\PPP(C=0|Z)=1-p$.
\end{prop}

\subsection{Optimal Way of Distinguishing Two Density Matrices}
\label{sec:rho1rho2}

The problem of distinguishing, at time $t_2$, whether the wave function is collapsed or not, can be regarded as a special case of the problem of distinguishing two probability distributions $\mu_1,\mu_2$ on $\SSS(\Hilbert)$, or of distinguishing two density matrices $\rho_1,\rho_2$. That is, suppose that a number $X\in\{1,2\}$ gets chosen randomly with $\PPP(X=1)=p$ and $\PPP(X=2)=1-p$; as a second step, we are given a system with a random wave function $\psi'$ chosen according to $\mu_X$. 
We can also suppose more directly that, as the second step in the experiment, we are given a system with density matrix $\rho_X$; this density matrix may arise from some distribution $\mu_X$, or from tracing out some other system, or both. 
Our previous scenario of \eqref{psiprimedef1}--\eqref{psiprimedef2} corresponds to the special case
\begin{align}\label{rho1}
\mu_1=\sum_k |\scp{b_k}{\psi}|^2 \delta_{c_kb_k}\,,
\quad\rho_1&=\diag \pr{\psi}\,,\\
\label{rho2}
\mu_2=\delta_\psi\,, \quad \rho_2&=|\psi\rangle\langle\psi|
\end{align}
with $c_k=\scp{b_k}{\psi}/|\scp{b_k}{\psi}|$ and
\be\label{diagdef}
\diag A = \sum_k \pr{b_k}A\pr{b_k}\,.
\ee

We now ask, how well can we retrodict $X$ from experiments on the system? Again, any experiment $\E$ with two possible outcomes, 1 and 2, is characterized by a self-adjoint operator $0\leq E\leq I$ (i.e., one with eigenvalues between 0 and 1), namely $E=E_1$ (so that $E_2=I-E$), and the usefulness of $\E$ for our purpose can be quantified by its reliability $R(\rho_1,\rho_2,\E)$, i.e., the probability that its outcome $Z$ agrees with $X$. Using that the main theorem about POVMs applies also to systems that have a density matrix, rather than a wave function, and then says that the outcome $Z$ has probability distribution
\be\label{mainthmPOVM2}
\PPP(Z=z)=\tr(\rho E)\,,
\ee
we find that
\begin{align}
R(\rho_1,\rho_2,\E) &= \PPP(Z=1,X=1) + \PPP(Z=2,X=2)\\
& = \PPP(Z=1|X=1)\PPP(X=1) + \PPP(Z=2|X=2)\PPP(X=2)\\
&= p \tr\bigl(E\rho_{1}\bigr) + (1-p) \tr\bigl((I-E)\rho_{2}\bigr)\\
&=1-p + \tr\Bigl( E\bigl( p\rho_{1} - (1-p) \rho_{2} \bigr)\Bigr)\,. \label{Rrho}
\end{align}
Note that the reliability depends on $\E$ only through $E$.

Which $E$ will maximize the reliability for given $\rho_1$, $\rho_2$, and $p$? 

\begin{prop}[Helstrom \cite{Hel76}]\label{prop:optimizer}
For a given self-adjoint operator $A$, $\tr(EA)$ is maximized among $E$s with $0\leq E\leq I$ by those and only those $E$ with
\be
P_+(A)\leq E \leq P_+(A)+P_0(A)\,,
\ee
where $P_+(A)$ is the projection to the subspace $\Hilbert_+(A)$ spanned by all eigenvectors of $A$ with positive eigenvalues, and $P_0(A)$ is the projection to the eigenspace $\Hilbert_0(A)$ of $A$ with eigenvalue 0. If 0 is not an eigenvalue of $A$, then $P_0(A)=0$, and the optimizer $E=P_+(A)$ is unique. The maximal value of $\tr(EA)$ is the sum of the positive eigenvalues of $A$ (with multiplicities). 
\end{prop}

In our case, the projection to $\Hilbert_+(p\rho_1 - (1-p)\rho_2)$ is an optimal choice of $E$, and the maximal reliability is $1-p$ plus the sum of all positive eigenvalues of $p\rho_1 - (1-p)\rho_2$; this reliability is $<1$ unless $\Hilbert_+(\rho_1)$ is orthogonal to $\Hilbert_+(\rho_2)$.

Coming back to the special case \eqref{rho1}--\eqref{rho2} with known $\psi$, this optimal $E$ operator and the maximal reliability can be specified explicitly. For the sake of completeness, we quote the formulas from \cite{CT12a}:
\be\label{Eopt}
E^{\mathrm{opt}}_{\psi,p} = \begin{cases}
I-\pr{\tilde\psi} & \text{if }0<p \leq \frac{n}{n+1}\,,\\
I & \text{if }\frac{n}{n+1}\leq p <1\,, \end{cases} 
\ee
with $\tilde\psi=M^{-1}\psi/\|M^{-1}\psi\|$, $M= z_{\psi,p} I +  \diag \pr{\psi}$, 
\be
z_{\psi,p} = f_\psi^{-1}\Bigl(\frac{p}{1-p}\Bigr)\geq 0\,, \quad
f_\psi(z) = \sum_{k=1}^n \frac{|\psi_k|^2}{z+|\psi_k|^2} \text{ for }z\geq 0\,,
\ee
and
\be
R^{\mathrm{opt}}_{\psi,p} = \begin{cases}  p(1+z_{\psi,p}) & \text{if } 0< p \leq \frac{n}{n+1}\,,\\
p & \text{if } \frac{n}{n+1}\leq p<1\,. \end{cases}
\ee

\subsection{If $\psi$ is Unknown}

Often it is desirable to have an experiment that works for unknown $\psi$. However, it is not obvious what it should mean for an experiment to work for unknown $\psi$. One approach, with a Bayesian flavor, would be to take this to mean that the experiment works (i.e., is more reliable than blind guessing) for random $\psi$ with uniform distribution. We have already discussed this scenario and can conclude that in this approach it is impossible for the inhabitants of a GRWm or GRWf world to measure the collapses.

However, it can be questioned whether an unknown $\psi$ can be assumed to be uniformly distributed. So here is another approach. Obviously, any experiment $\E$ that we choose will have high reliability for some $\psi$s and low reliability (lower than blind guessing) for other $\psi$s (as the average reliability over all $\psi$ equals that of blind guessing). We may wish to maximize the size of the set
\be
S_\E = \bigl\{ \psi\in\SSS(\Hilbert): R_\psi(\E)>\max\{p,1-p\} \bigr\}\,,
\ee
the set of $\psi$s for which $\E$ is more reliable than blind guessing. The natural measure of ``size'' is the (normalized) uniform distribution $u$ on $\SSS(\Hilbert)$. There do exist experiments for which $u(S_\E)>1/2$ \cite{CT12c}, but we have reason \cite{CT12c} to conjecture that $u(S_\E)\leq 1-1/e\approx 0.632$. If this is right, it is another curious limitation to knowledge: While you can do better than blind guessing on more than 50\%\ of all wave functions, you cannot do better than blind guessing on more than 64\%\ of all wave functions. 

Here are further results concerning $u(S_\E)$, expressing limitations to knowledge.

\begin{prop}\label{prop:uSE} \cite{CT12c}
For $p<1/2-1/\sqrt{8}\approx 0.146$ and $p>1/2+1/\sqrt{8}\approx 0.854$, all $\Hilbert$ and all $\E$, $u(S_\E)\leq 1/2$.
\end{prop}

That is, for $p$ close to 0 or 1, one cannot even do better than blind guessing for a majority of wave functions. 

\begin{prop} \cite{CT12c}
For $\Hilbert$ with $\dim\Hilbert=2$, all $0\leq p\leq 1$ and all $\E$, $u(S_\E)\leq 1/2$.
\end{prop}

\section{Measurements of $m(x)$ in GRWm}
\label{sec:m}

In this section we investigate the accuracy and reliability of genuine measurements of $m(x)$ by inhabitants of a GRWm universe. We have shown above that $m(x)$ is not measurable with microscopic accuracy. We now show that it \emph{is} measurable on the macroscopic level. Our analysis can be regarded as an elaboration of statements by Ghirardi \emph{et al.}\ \cite{GG96}, \cite[section 10.2]{BG03} to the effect that different degrees of ``accessibility'' of $m(x)$ can occur. Specifically, we confirm that the quantity $\R(V)$ that they introduced governs the measurability of $m(x)$: the average of $m(x)$ over $V$ can be measured accurately and reliably whenever $\R(V)$ is small.

So consider a coarse-grained, macroscopic version $\mm(x)$ of $m(x)$, for example
\be
\mm_1(x) = (g_\ell\ast m)(x) = \sum_{i=1}^N m_i \, \scp{\psi}{g_\ell(\widehat Q_i-x)|\psi}
\ee
with $\ast$ convolution, $g_\ell$ the Gaussian function of width $\ell$ as in \eqref{Gaussian}, and $\ell$ the length scale of the coarse graining (independent of the GRW length $\sigma$), or, based on a partition of $\RRR^3$ into cubes of side length $\ell$, 
\be
\mm_2(x) = \frac{1}{\ell^3} \int\limits_{C(x)} dx'\, m(x')\,,
\ee 
with $C(x)$ the cube containing $x$. 

Here is a simple procedure for a measurement of $\mm(x)$ (or of $m(x)$ with inaccuracy $\ell$). While this procedure is not practically feasible, it shows the possibility in principle and may suggest more practical schemes if those are desired. We consider a macroscopic system of $N$ ``particles'' with (GRW) wave function $\psi(q)=\psi(q_1,\ldots,q_N)$. Carry out an ideal quantum measurement of the position observables on all $N$ particles, with outcome $Q=(Q_1,\ldots,Q_N)$ distributed with distribution density $\PPP(Q=q)=|\psi(q)|^2$; let the estimator for $\mm(x)$ be
\be
M(x) = \sum_{i=1}^N m_i \, \delta^3(Q_i-x)
\ee
or a coarse-grained version $\MM(x)$ of $M(x)$, for example
\be
\MM_1(x) = \sum_{i=1}^N m_i \, g_\ell(Q_i-x)
\ee
or
\be
\MM_2(x) = \frac{1}{\ell^3} \int\limits_{C(x)} dx'\, M(x') = \frac{1}{\ell^3}\sum_{i:Q_i\in C(x)} m_i\,.
\ee 
Put differently, $M(x)$ are the outcomes of a joint ideal quantum measurement of the commuting observables $\widehat{M}(x)$, the mass density operators as in \eqref{hatMdef}.

It follows from \eqref{mdef2} that
\be
m(x)=\EEE\, M(x)\text{ and }\mm_{1,2}(x) = \EEE \MM_{1,2}(x)\,.
\ee
So the estimator will be close to the true value if its variance is small. Specifically, the relative inaccuracy with which $\mm_2(x)$ can be measured in this way is the standard deviation of $\MM_2(x)$ divided by $\mm_2(x)$, which is exactly Ghirardi's $\R(V)$ with $V=C(x)$,
\be\label{RCx}
\R(C(x))= \frac{\scp{\psi}{(\tfrac{1}{\ell^3}\widehat{M}(C(x))-\mm_2(x))^2|\psi}^{1/2}}{\mm_2(x)}\,.
\ee
While we do not have a proof that no other method of estimating $\mm_2(x)$ is more accurate, this seems plausible, as no better method comes to mind.

Obviously, the inaccuracy \eqref{RCx} depends on the wave function $\psi$. This leads us to the question, for typical wave functions arising from the GRW process, how large do we have to choose $\ell$ to make the inaccuracy $\R(C(x))$ smaller than, say, 10\%? Let us consider a few examples.

An object of macroscopic size consisting of a uniform solid or liquid at everyday conditions has interatomic distances of $10^{-10}$ to $10^{-9}$ m, and we may expect the wave function of the nucleus to be spread out over similar distances (except for permutation symmetry). Since, for $\ell=3\times 10^{-9}$ m, a volume of $\ell^3$ contains 30 to $3\times 10^4$ atoms, $\MM_2(x)$ involves an average over many atoms; thus, for this size of $\ell$ or larger, we may expect the variation of the wave function (such as the ground state of the solid) to have little effect on $\MM_2(x)$, and $\R(C(x))$ to be small.

Now consider a solid object $O$ of size $\delta$ (or even a sheet of thickness $\delta$, since only 1 dimension of space is relevant), an $\ell>\delta$, and a $\psi$ that is a superposition of two different positions of $O$ with a difference greater than $\ell$. Then one term in the superposition may correspond to $O$ lying entirely in $C(x)$, while the other term corresponds to $O$ lying entirely outside of $C(x)$. Let the coefficients of the terms be $c_1$ and $c_0$, respectively; then, for a suitable constant $m_0$, $\MM_2(x)$, if measured, equals $m_0$ with probability $p=|c_1|^2$ and equals $0$ with probability $q=|c_2|^2=1-|c_1|^2$. It follows that $\mm_2(x)=pm_0$, that the standard deviation of $\MM_2(x)$ is $m_0\sqrt{pq}$, and that $\R(C(x))=\sqrt{q/p}$, which is greater than $10\%$ for every $p<99\%$. Thus, except for extreme values of $p$, such a superposition yields quite a large value of $\R(C(x))$, and thus probably low accuracy as a genuine measurement of $\mm_2(x)$. However, since superpositions of different locations more than $\sigma$ apart are suppressed by the spontaneous collapses, such wave functions are unlikely to occur for $\ell\gg \sigma=10^{-7}$ m. Or rather, they cannot persist much longer than for $\Delta t= 1/N\lambda$.

These considerations suggest that $m(x)$ can usually be measured with small relative inaccuracy and high reliability at a spatial and temporal resolution of
\be
\Delta x = 10^{-7}\text{ m}\,,\quad
\Delta t = \frac{1}{N_{\Delta x}\lambda}\,,
\ee
where $N_{\Delta x}$ is the number of nucleons in the volume $(\Delta x)^3$. On the other hand, examples can be set up, at least artificial ones, for which these estimates are not valid: if $\psi=\varphi^{\otimes N}$ with $\varphi$ a 1-particle wave function that is spread out over distances much greater than $\sigma$, and if the $N$ particles do not interact (say, $H=0$), then the spontaneous collapses are not efficient at localizing the wave function, and $\R(C(x))$ will not become small until after $10^8$ years (when almost every particle has been hit by a collapse).

\appendix

\section{Proofs}
\label{sec:proofs}

\begin{prop}\label{prop:ens}
If $\mu_1\neq \mu_2$ are distributions of wave functions with equal density matrices $\rho_{\mu_1}=\rho_{\mu_2}$ then no experiment can distinguish between an ensemble of wave functions with distribution $\mu_1$ and one with $\mu_2$.
\end{prop}

\begin{proof}
This is a consequence of the main theorem about POVMs \eqref{mainthmPOVM}. If $\E$ is carried out on an ensemble of systems with wave functions $\psi\in\SSS(\Hilbert_\sys)$ distributed according to $\mu_1$ then the probability of obtaining outcome $z$ is
\be
\PPP(Z=z) 
=  \tr\Bigl( E(z) \rho_{\mu_1} \Bigr)\,,
\ee
so it depends on $\mu_1$ only through $\rho_{\mu_1}$. Thus, since $\rho_{\mu_2}=\rho_{\mu_1}$, the distribution of outcomes would be the same in an ensemble distributed according to $\mu_2$.
\end{proof}

\begin{prop}\label{prop:reliability}
Consider the collapse process as in \eqref{psiprimedef1}--\eqref{psiprimedef2} and arbitrary $\E$. Let $E_z$ be the POVM associated with $\E$ by the main theorem about POVMs, and $E=E_\yes$. Then 
\be\label{RpsiA}
R_\psi(\E) = \scp{\psi}{A|\psi} 
\ee
with
\be
A=p\diag E + (1-p)(I-E)
\ee
and $\diag$ the ``diagonal part'' as defined in \eqref{diagdef}. For $\psi$ as in \eqref{superposition}, $R_\psi(\E_1)=1-p/2$ (i.e., Eq.~\eqref{RE1} holds). For random $\psi$ with distribution $\mu$,
\be\label{RmuA}
R_\mu(\E) = \tr(\rho_\mu \, A)\,.
\ee
\end{prop}

\begin{proof}
We have that
\begin{align}
\PPP(Z=\yes|C=0) &= \scp{\psi}{E|\psi}\\[4mm] 
\PPP(Z=\yes|C=1) &= \sum_k \bigl| \scp{b_k}{\psi} \bigr|^2\,\scp{b_k}{E|b_k} \\
&= \scp{\psi}{\diag E|\psi}\,.
\end{align}
Now \eqref{RpsiA} follows, together with \eqref{RE1} and \eqref{RmuA}; \eqref{RE1} is obtained as the special case with $\psi=2^{-1/2}(b_1+b_2)$ and $E=I-\pr{\psi}$. 
\end{proof}

\noindent{\bf Proposition~\ref{prop:uniform}.} 
\textit{Let $u$ be the uniform distribution on $\SSS(\Hilbert)$. Then
$R_u(\E)\leq \max\{p,1-p\}$ for every experiment $\E$ at time $t_2$. If $p=\frac{1}{2}$ then, in fact, $R_u(\E)=1/2$ for every experiment $\E$ at time $t_2$.}

\medskip

\begin{proof}
By \eqref{RmuA},
\be\label{RuE}
R_u(\E) = \frac{1}{n} \tr A = \frac{p}{n} \tr  E + \frac{1-p}{n}\tr(I-E)= 1-p-\frac{1-2p}{n}\tr E\,. 
\ee
Note that $0\leq E\leq I$ and thus $0\leq \tr E \leq n$. For $0\leq p\leq 1/2$, \eqref{RuE} is $\leq 1-p$ since $\tr E\geq 0$ and $1-2p\geq 0$. For $1/2<p\leq 1$, rewrite \eqref{RuE} as
\be
R_u(\E) = p-\frac{2p-1}{n}\tr(I- E) \,,
\ee
note $\tr(I-E)\geq 0$ and $2p-1\geq 0$, and conclude $R_u(\E)\leq p$.
\end{proof}

\noindent{\bf Proposition~\ref{prop:noinf}.} 
\textit{Consider the collapse process as in \eqref{psiprimedef1}--\eqref{psiprimedef2} and an arbitrary experiment $\E$ at time $t_2$, possibly with more than two possible outcomes. Let $\psi$ be uniformly distributed on $\SSS(\Hilbert)$. Then $\PPP(C=1|Z)=p$ and $\PPP(C=0|Z)=1-p$.}

\medskip

\begin{proof}
Let $E_z$ be the POVM of $\E$; then
\begin{align}
\PPP(C=1|Z=z)
&= \frac{\PPP(Z=z|C=1)\PPP(C=1)}{\PPP(Z=z|C=0)\PPP(C=0)+\PPP(Z=z|C=1)\PPP(C=1)}\\[6mm]
&=\frac{\tr(\rho\,\diag E_z) p }{\tr (\rho\, E_z)(1-p)+\tr(\rho\, \diag E_z)p}=p
\end{align}
using $\rho=\frac{1}{n}I$.
\end{proof}

\bigskip

\textit{Acknowledgments.}
We thank Sheldon Goldstein for helpful discussions.
Both authors are supported in part by NSF Grant SES-0957568. 
R.T.~is supported in part by grant no.~37433 from the John Templeton Foundation and by the Trustees Research Fellowship Program at Rutgers, the State University of New Jersey. 

\end{document}